\documentclass{article}
\usepackage{amssymb,amsmath,amsthm}
\usepackage{newcent}
\usepackage{color,graphicx}
\usepackage{epstopdf}
\newcommand{\comment}[1]{} 

\def\d{{\rm d}}

\def\D{\mbox{\rm D}} 
\def\i{\mbox{\rm i}}
\def\cdot{{\scriptstyle\,\bullet\,}}

\def\p{Pain\-lev\'{e}}
\def\kp{Kadomtsev-Petviashvili}
\def\sch{Schr\"{o}dinger}
\def\PI{\mbox{\rm P$_{\rm I}$}}
\def\PII{\mbox{\rm P$_{\rm II}$}}

\def\PIV{\mbox{\rm P$_{\rm IV}$}}

\newcommand{\pderiv}[3][]{\frac{\partial^{#1}{#2}}{{\partial{#3}}^{#1}}}

\def\beq{\begin{equation}}
\def\eeq{\end{equation}}

\def\imp{\int_{-\infty}^{\infty}}
\def\intR{\int_{-R}^{R}}
\def\dimp{\imp\imp}
\def\dintR{\intR\intR}

\def\O{\mathcal{O}}

\def\a{\alpha}
\def\b{\beta}

\def\ph{\varphi}
\def\ep{\varepsilon}
\def\etal{\textit{et al.}}
\def\i{\ifmmode{\rm i}\else\char"10\fi}
\def\uu{u}
\def\ff{F}

\newtheorem{theorem}{Theorem}
\newtheorem{corollary}[theorem]{Corollary}
\newtheorem{remark}[theorem]{Remark}

\newtheorem{conjecture}[theorem]{Conjecture}
\newtheorem{lemma}[theorem]{Lemma}
\numberwithin{equation}{section}
\numberwithin{table}{section}
\numberwithin{figure}{section}
\numberwithin{theorem}{section}

\begin{document}
\title{Conservation laws and integral relations\\ 
for the Boussinesq equation}
\author{Adrian Ankiewicz\\
Optical Sciences Group, Research School of Physics and Engineering,\\
The Australian National University, Canberra ACT 2601, Australia\\
Email: \texttt{ana124@physics.anu.edu.au}\\[5pt]
Andrew P.\ Bassom\\
School of Mathematics and Physics, University of Tasmania,\\ Private Bag 37, Hobart, Tasmania 7001, Australia\\
Email: \texttt{andrew.bassom@utas.edu.au}\\[5pt]
Peter A.\ Clarkson and Ellen Dowie\\ School of Mathematics, Statistics and Actuarial Science\\
University of Kent, Canterbury, CT2 7FS, UK\\
Email: \texttt{P.A.Clarkson@kent.ac.uk}, \texttt{ed275@kent.ac.uk}
}

\maketitle
\begin{abstract}
We are concerned with conservation laws and integral relations associated with rational solutions of the Boussinesq equation, a soliton equation solvable by inverse scattering which was first introduced by Boussinesq in 1871. The rational solutions are logarithmic derivatives of a polynomial, are algebraically decaying and have a similar appearance to rogue-wave solutions of the focusing nonlinear Schr\"{o}dinger equation. For these rational solutions the constants of motion associated with the conserved quantities are zero and they have some interesting  integral relations which depend on the total degree of the associated polynomial.
\end{abstract}

\begin{center}\textit{This paper is dedicated to the memory of Professor David J.\ Benney}\end{center}

\section{\label{sec1}Introduction} 
In this paper we discuss conservation laws and integral relations associated with algebraically decaying rational solutions $u=u(x,t)$ of the Boussinesq equation 
\beq
u_{tt}+u_{xx}-(u^2)_{xx}-\tfrac13u_{xxxx}=0,\label{eq:bq}
\eeq
where subscripts denote partial derivatives. Equation \eqref{eq:bq} was introduced by Boussinesq in 1871 to describe the propagation of long waves in shallow water \cite{refBousa,refBousc}; see also \cite{refUrs,refWhit}. 
Benney and Luke \cite{refBL64} showed that certain classical equations derived by mathematicians in the late 1800s, such as the Boussinesq equation \eqref{eq:bq} and the Korteweg-de Vries equation
\beq u_t+6uu_x+u_{xxx}=0,\label{eq:kdv}\eeq
actually were generic approximations of weakly nonlinear-weakly dispersive wave phenomena.
The Boussinesq equation \eqref{eq:bq} is also a soliton equation solvable by inverse scattering \cite{refAH,refBZ,refDTT,refZak} which arises in several other physical applications including one-dimensional nonlinear lattice-waves \cite{refToda,refZab}, vibrations in a nonlinear string \cite{refZak}, and ion sound waves in a plasma \cite{refInR,refScott}. We remark that equation \eqref{eq:bq} is sometimes referred to as the ``bad" Boussinesq equation, i.e.\ when the ratio of the $u_{tt}$ and $u_{xxxx}$ terms is negative. If the sign of the $u_{xxxx}$ term is reversed in \eqref{eq:bq}, then the equation is sometimes called the ``good" Boussinesq equation. The coefficients of the $u_{xx}$ and $(u^2)_{xx}$ terms can be changed by scaling and translation of the dependent variable $u$. For example, letting $u\to u+1$ in \eqref{eq:bq} gives
\beq
u_{tt}-u_{xx}-(u^2)_{xx}-\tfrac13u_{xxxx}=0,\label{eq:bq1}
\eeq 
which is the non-dimensionalised form of the equation originally written down by Boussinesq  \cite{refBousa,refBousc}.

Recently Clarkson and Dowie \cite{refCD16} studied rational solutions $u_n(x,t)$ of the Boussinesq equation \eqref{eq:bq}. These rational solutions, which are algebraically decaying and can depend on two arbitrary parameters, have the form
\beq\label{ratun} u_n(x,t;\a,\b) = 2\pderiv[2]{}{x}\ln F_n(x,t;\a,\b),\eeq
where $F_{n}(x,t;\a,\b)$ is a polynomial of degree $n(n+1)$ in both $x$ and $t$, with total degree $n(n+1)$, with $\a$ and $\b$ parameters. The polynomial $F_{n}(x,t;\a,\b)$ satisfies a fourth-order, bilinear equation -- see \eqref{eq:bilin} below.
These rational solutions have a similar appearance to rogue-wave solutions of the focusing nonlinear \sch\ (NLS) equation, 
cf.~\cite{refAASG,refACA,refAKA,refKAA11,refKAA12,refKAA13}
\begin{equation}
\mbox{\rm i} \psi_{t} + \psi_{xx}+\tfrac1{2}|\psi|^2 \psi =0,
\label{eq:fnls}
\end{equation} 
which also is a soliton equation solvable by inverse scattering \cite{refZS72}.
Benney and Newell \cite{refBN67} showed that the NLS equation arises universally in diverse applications in nonlinear dispersive waves.

In \S\ref{sec2}, we review the results in \cite{refCD16} concerning algebraically decaying rational solutions $u_n(x,t;\a,\b)$ of the Boussinesq equation \eqref{eq:bq}. In \S\ref{sec3}, we discuss conservation laws for the Boussinesq equation \eqref{eq:bq}, in particular showing that the constants of motion for the rational solutions $u_n(x,t;\a,\b)$ given by \eqref{ratun} are all zero.
In \S\ref{sec4}, we discuss integral relations for these rational solutions of the Boussinesq equation \eqref{eq:bq}. Specifically we prove the following result:
\begin{theorem}\label{thm1} Suppose that $u_n(x,t;\a,\b)$ is an algebraically decaying rational solution of the Boussinesq equation \eqref{eq:bq}
of the form \eqref{ratun}, then
\begin{align} \frac{1}{8\pi}\dimp u_n^2(x,t;\a,\b)\,\d x\,\d t &= \tfrac12n(n+1),\label{int:un2} \\ 
\frac{1}{8\pi}\dimp u_n^3(x,t;\a,\b)\,\d x\,\d t &= n(n+1).\label{int:un3}\end{align}
\end{theorem} 
Theorem \ref{thm1} shows a relationship between the integrals and the total degree of the polynomial $F_n(x,t;\a,\b)$ associated with the rational solution $u_n(x,t;\a,\b)$. An analogous result to \eqref{int:un2} has been conjectured for rogue-wave solutions of the NLS equation \eqref{eq:fnls} \cite{refAA15}.
In \S\ref{seckp} we discuss rational solutions of the \kp\ I (KPI) equation 
\begin{equation}\label{eq:kp1}
(v_\tau+6vv_\xi+v_{\xi\xi\xi})_\xi=3 v_{\eta\eta},
\end{equation}
specifically how the approach of Ablowitz \etal\ \cite{refACTV,refAV,refVA99} might elucidate the results obtained here.
In \S\ref{sec5} we discuss our results.

\section{\label{sec2}Rational solutions of the Boussinesq equation}
Clarkson and Kruskal \cite{refCK} showed that Boussinesq equation \eqref{eq:bq} has symmetry reductions to 
the first, second and fourth \p\ equations\ (\PI, \PII, \PIV).
Since \PII\ and \PIV\ themselves have rational solutions, symmetry reductions were used in \cite{refPAC08} to derive rational solutions of the Boussinesq equation \eqref{eq:bq}. Further more general rational solutions of  
\eqref{eq:bq} are also given in \cite{refPAC08}.
Unfortunately none of these rational solutions are bounded for all real $x$ and $t$, and so it is unlikely that they will have any physical significance.

However it is known that there are additional rational solutions of the Boussinesq equation \eqref{eq:bq} which do not arise from the above construction. For example, Ablowitz and Satsuma \cite{refAbSat} derived the rational solution
\beq\label{bq:rat1}
u(x,t)= 2\pderiv[2]{}{x}\ln(1+x^2+t^2)=\frac{4(1-x^2+t^2)}{(1+x^2+t^2)^2} ,
\eeq 
by taking a long-wave limit of the two-soliton solution, see also \cite{refTajMur91,refTajWat}. This solution is bounded for real $x$ and $t$, and tends to zero algebraically as $|x|\to\infty$ and $|t|\to\infty$. 

If in the Boussinesq equation \eqref{eq:bq}, we make the transformation \beq u(x,t)=2\pderiv[2]{}{x}\ln F(x,t),\eeq
then we obtain the bilinear equation \cite{refHirota73,refHirSat76}
\beq FF_{tt}-F_t^2 + FF_{xx}-F_x^2 - \tfrac13\left(FF_{xxxx} -4 F_xF_{xxx}+3F_{xx}^2\right)=0,\label{eq:bilin}
\eeq
which can be written in the form
\beq(\D_t^2 + \D_x^2-\tfrac13\D_x^4)F\cdot F=0,\eeq
where $\D_x$ and $\D_t$ are Hirota operators 
\begin{align}\label{hirotaop}
\D_x^m \,\D_t^n F(x,t)\cdot F(x,t)&=\left[\left(\pderiv{}{x}-\pderiv{}{x'}\right)^m\left(\pderiv{}{t}-\pderiv{}{t'}\right)^n F(x,t)F(x',t')\right]_{x'=x,t'=t}.
\end{align}

\def\fig#1{\includegraphics[width=5cm]{#1}}
\begin{figure}[ht]
\[ \begin{array}{c@{\quad}c@{\quad}c} 
\fig{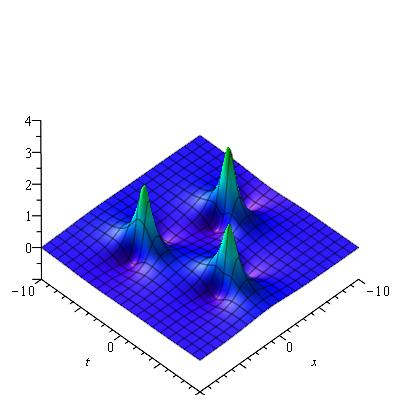}& \fig{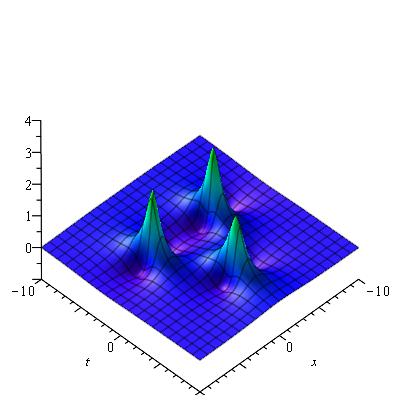}& \fig{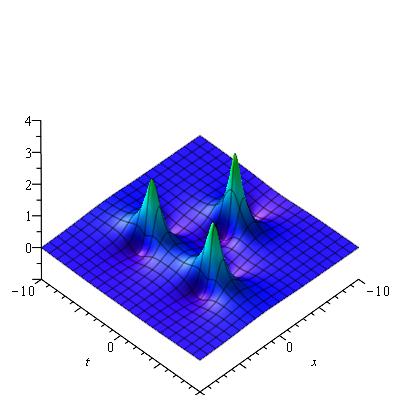}\\
\a=\b=100 &  \a=100, \b=0  & \a=0, \b=100 \\
\fig{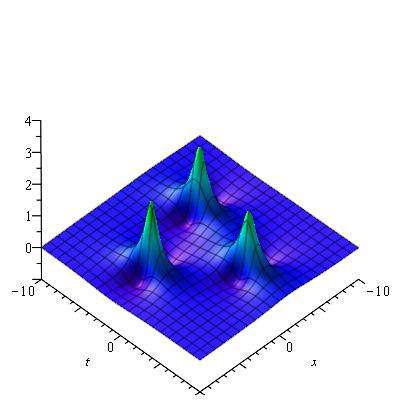}& \fig{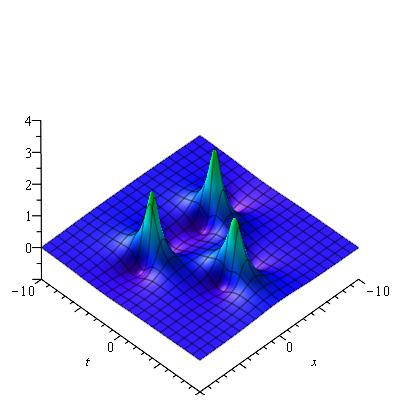}& \fig{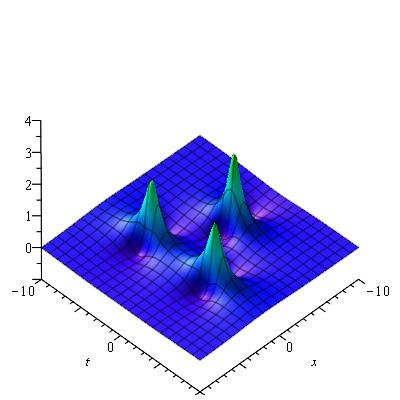}\\
\a=100, \b=-100 &  \a=100, \b=10  & \a=10, \b=100 
\end{array}\]
\caption{\label{fig21}Plots of the solution $\uu_2(x,t;\a,\b)$ for various choices of $\a$ and $\b$.}
\end{figure}

\begin{figure}[ht]
\[ \begin{array}{c@{\quad}c@{\quad}c} 
\fig{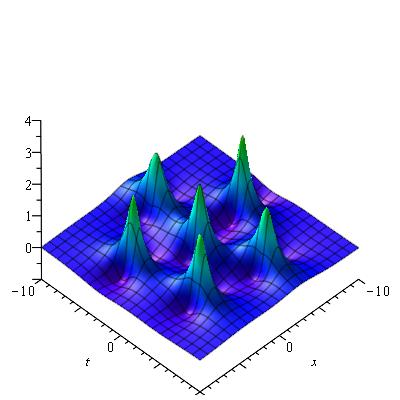}& \fig{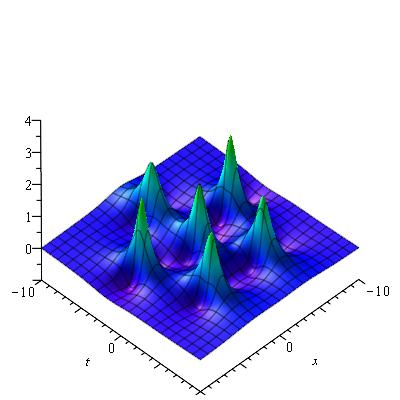}& \fig{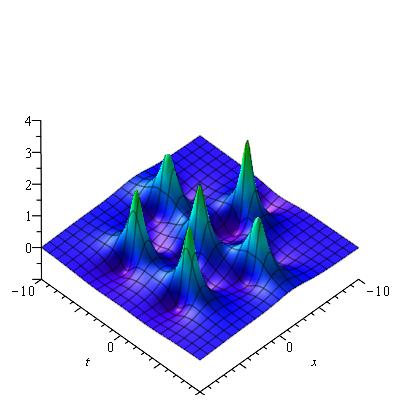}\\
\a=\b=10^4 &  \a=10^4, \b=0  & \a=0, \b=10^4 \\
\fig{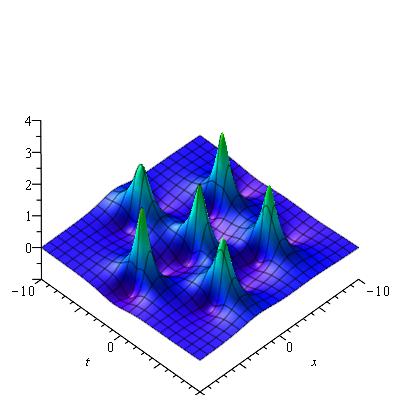}& \fig{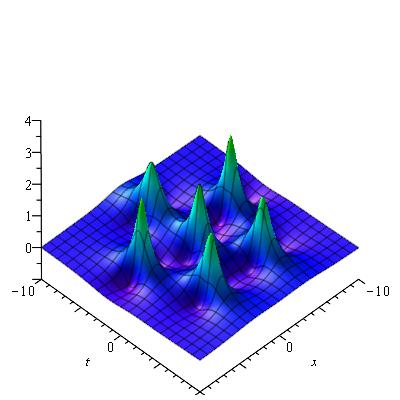}& \fig{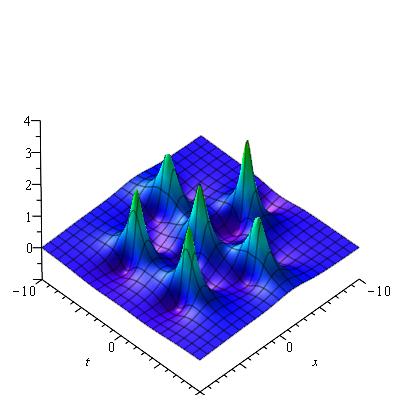}\\
\a=10^4, \b=-10^4  & \a=10^4, \b=10^2 &  \a=10^2, \b=10^4
\end{array}\]
\caption{\label{fig22}Plots of the solution $\uu_3(x,t;\a,\b)$ for various choices of $\a$ and $\b$.}
\end{figure}

Since the Boussinesq equation \eqref{eq:bq} admits the rational solution \eqref{bq:rat1}, Clarkson and Dowie \cite{refCD16} sought solutions in the form
\beq\label{rat:ansatz1}
{u}_n(x,t) = 2\pderiv[2]{}{x}\ln f_n(x,t),\qquad n\geq1,\eeq
where $F_{n}(x,t)$ is a polynomial of degree $n(n+1)$ in $x$ and $t$, with total degree $n(n+1)$. In particular
\beq\label{rat:ansatz2}
f_n(x,t) = \sum_{m=0}^{n(n+1)/2}\sum_{j=0}^ma_{j,m}x^{2j}t^{2(m-j)},
\eeq
with $a_{j,m}$ constants, which are determined using \eqref{eq:bilin} and equating powers of $x$ and $t$. Using this ansatz the following polynomials were obtained:
\begin{subequations}\label{bq:fn}\begin{align} 
f_1(x,t)&=x^2+t^2+1,\\
f_2(x,t)&={x}^{6}+ \left( 3{t}^{2}+\tfrac{25}3 \right) {x}^{4}+ \left( 3{t}^{4}+30{t}^{2}-\tfrac{125}9 \right) {x}^{2}+{t}^{6}+\tfrac{17}3{t}^{4}+\tfrac{475}9{t}^{2}+\tfrac{625}9,\\
f_3(x,t)&={x}^{12}+ \left( 6{t}^{2}+\tfrac{98}3 \right) {x}^{10}+ \left( 15{t}^{4}+230{t}^{2}+\tfrac{245}3 \right) {x}^{8}
+ \left( 20{t}^{6}+\tfrac{1540}3{t}^{4}+\tfrac{18620}9{t}^{2}+{\tfrac {75460}{81}} \right) {x}^{6}\nonumber\\
&\phantom{={x}^{12}\ }+ \left( 15{t}^{8}+\tfrac{1460}3{t}^{6}+\tfrac{37450}9{t}^{4}+\tfrac{24500}3{t}^{2}-{\tfrac {5187875}{243}}\right) {x}^{4} \nonumber\\
&\phantom{={x}^{12}\ }+ \left( 6{t}^{10}+190{t}^{8}+\tfrac{35420}9{t}^{6}-\tfrac{4900}9{t}^{4}+\tfrac{188650}{27}{t}^{2}+{\tfrac {159786550}{729}} \right) {x}^{2}\nonumber\\ 
&\phantom{={x}^{12}\ }+{t}^{12}+\tfrac{58}3{t}^{10}+\tfrac{1445}3{t}^{8}+{\tfrac {798980}{81}}{t}^{6}+{\tfrac {16391725}{243}}{t}^{4}
+{\tfrac {300896750}{729}}{t}^{2}+{\tfrac {878826025}{6561}}.
\end{align}\end{subequations}
(The lengthy polynomials $f_4(x,t)$ and $f_5(x,t)$ are also given in \cite{refCD16}.)

Clarkson and Dowie \cite{refCD16} further showed that the Boussinesq equation \eqref{eq:bq} possesses generalised rational solutions of the form
\beq \uu_n(x,t;\a,\b) = 2\pderiv[2]{}{x}\ln \ff_n(x,t;\a,\b),\label{ungen}\eeq
for $n\geq1$, with $\a$ and $\b$ arbitrary constants, where
\beq\label{bq:untilde}
\ff_{n+1}(x,t;\a,\b)= f_{n+1}(x,t) +2\a tP_{n}(x,t)+2\b xQ_{n}(x,t)+\big(\a^2+\b^2\big)f_{n-1}(x,t),\eeq
with $f_n(x,t)$ given by \eqref{bq:fn}, and
$P_{n}(x,t)$ and $Q_{n}(x,t)$ polynomials of degree $n(n+1)$ in $x$ and $t$. Specifically
\begin{align}\label{bq:PQform}
P_n(x,t) &= \sum_{m=0}^{n(n+1)/2}\sum_{j=0}^m b_{j,m}x^{2j}t^{2(m-j)},\qquad
Q_n(x,t) = \sum_{m=0}^{n(n+1)/2}\sum_{j=0}^m c_{j,m}x^{2j}t^{2(m-j)},
\end{align}
with the constants $b_{j,m}$ and $c_{j,m}$ determined by equating powers of $x$ and $t$.
By substituting \eqref{bq:untilde} into the bilinear equation \eqref{eq:bilin} with $f_j(x,t)$, for $j=1,2,\ldots,5$, given by \eqref{bq:fn}, then 
it is shown in \cite{refCD16} that
\begin{subequations}\label{bq:PQ}\begin{align}
P_1(x,t)&=3x^2-t^2+\tfrac53,\\
Q_1(x,t)&=x^2-3 t^2-\tfrac13,\\ 
P_2(x,t)&=5x^6 -\left(5t^2-35\right)x^4-\left(9t^4+\tfrac{190}3t^2+\tfrac{665}{9}\right)x^2+t^6-\tfrac73t^4-\tfrac{245}9t^2+\tfrac{18865}{81},\\
Q_2(x,t)&=x^6-\left(9t^2-\tfrac{13}3\right)x^4-\left(5t^4+\tfrac{230}3t^2+\tfrac{245}9\right)x^2 +5t^6+15t^4+\tfrac{535}9t^2+\tfrac{12005}{81},
\end{align}\end{subequations}
with $\a$ and $\b$ arbitrary constants; the polynomials $P_3(x,t)$, $Q_3(x,t)$, $P_4(x,t)$ and $Q_4(x,t)$ are given in \cite{refCD16}.
The polynomials have the form
\beq F_n(x,t;\a,\b)= \big(x^2+t^2\big)^{n(n+1)/2} + G_{n}(x,t;\a,\b),\label{FnGn}\eeq
where $G_n(x,t;\a,\b)$ is a polynomial of degree $(n+2)(n-1)$ in both $x$ and $t$. 
In Figures \ref{fig21} and \ref{fig22}, plots of $\uu_2(x,t;\a,\b)$ and $\uu_3(x,t;\a,\b)$ are given for various choices of $\a$ and $\b$, respectively.

\section{\label{sec3}Conservation laws} 
A conservation law is an equation of the form 
\beq\pderiv{T}{t} + \pderiv{X}{x}= 0,\label{conlaw}\eeq
where $T(x,t)$ is the \textit{conserved density} and $X(x,t)$ the \textit{associated flux}. The integral \beq\imp T(x,t)\,\d x=c,\label{com1}\eeq with $c$ a constant, is called a \textit{constant of motion}, with $t$ interpreted as a timelike variable. It follows that \beq\imp X(x,t)\,\d t=k,\label{com2}\eeq with $k$ also a constant.

In order to study conservation laws for the Boussinesq equation \eqref{eq:bq} we cast it as the system
\begin{subequations}\label{eq:bqs}\begin{align}
u_t&+v_x=0,\label{eq:bqsa}\\
v_t&+(u^2)_{x}-u_{x}+\tfrac13 u_{xxx}=0.\label{eq:bqsb}
\end{align}\end{subequations}
If $u(x,t)$ has the form $u(x,t)=2[\ln F(x,t)]_{xx}$, 
where $F(x,t)$ satisfies the bilinear equation \eqref{eq:bilin} then equation \eqref{eq:bqsa}, shows that 
\beq v(x,t)=-2\frac{\partial^2}{\partial x\partial t}\ln F(x,t).\eeq
The first few conserved densities $T_j(x,t)$ and associated fluxes $X_j(x,t)$ for the system \eqref{eq:bqs} are 
\begin{align*}
T_1 (x,t)&= u, && X_1(x,t)=v,\\
T_2(x,t)&=v, && X_2(x,t)=u^2-u+\tfrac13 u_{xx},\\
T_3(x,t)&= uv, && X_3(x,t)= \tfrac23 u^3+\tfrac12 v^2-\tfrac12 u^2-\tfrac16 u_x^2+\tfrac13 uu_{xx},\\
T_4(x,t)&=\tfrac23 u^3+v^2-u^2-\tfrac13u_x^2, &&
X_4(x,t)=2u^2v-2uv+\tfrac23 vu_{xx}-\tfrac23 u_xv_x;
\end{align*}
see Hereman \etal\ \cite{refHAEHH} for details.
Hence the first few constants of the motion for the system \eqref{eq:bqs} are
\begin{subequations}\label{bqcom1}\begin{align} 
&\imp u(x,t)\,\d x=c_1, \label{bqcom11}\\
&\imp v(x,t)\,\d x=c_2,\label{bqcom12}\\ 
&\imp u(x,t)v(x,t)\,\d x=c_3,\label{bqcom13}\\
&\imp \left(\tfrac23 u^3+v^2-u^2-\tfrac13u_x^2\right)\d x=c_4,\label{bqcom14}
\end{align}\end{subequations}
where $c_1$, $c_2$, $c_3$ and $c_4$ are constants.
The integral \eqref{bqcom13} corresponds to the conservation of momentum and \eqref{bqcom14} to the conservation of energy.
Further from the associated fluxes we have 
\begin{subequations}\label{bqcom2}\begin{align} 
&\imp v(x,t)\,\d t=k_1,\label{bqcom21}\\
& \imp \left(u^2-u+\tfrac13 u_{xx}\right)\d t=k_2,\label{bqcom22}\\ 
&\imp \left( \tfrac23 u^3+\tfrac12 v^2-\tfrac12 u^2-\tfrac16 u_x^2+\tfrac13 uu_{xx}\right)\d t=k_3,\label{bqcom23}\\
&\imp \left(2u^2v-2uv+\tfrac23 vu_{xx}-\tfrac23 u_xv_x\right)\d t=k_4,\label{bqcom24}
\end{align}\end{subequations}
where $k_1$, $k_2$, $k_3$ and $k_4$ are constants.

It is easily shown that for the algebraically decaying rational solutions of the Boussinesq equation \eqref{eq:bq} described in \S\ref{sec2}, then $c_j=0$ and $k_j=0$, for $j=1,\ldots,4$. 

\section{\label{sec4}Integral relations of rational solutions}
In this section we examine Theorem \ref{thm1}. The results hold for the generalised rational solutions $\uu_n(x,t;\a,\b)$ given by \eqref{ungen}, though in this section we will suppress the explicit dependence of the parameters $\a$ and $\b$.

\subsection{Integral of $u_n^2(x,t)$.}
First we shall consider the integral of $u_n^2(x,t)$, i.e.\ result \eqref{int:un2}.
Setting $u=U_{xx}$ in the Boussinesq equation \eqref{eq:bq} and then integrating twice w.r.t.\ $x$, assuming that $u$ and its derivatives vanish sufficiently rapidly as $|x|\to\infty$, yields
\beq u^2=U_{tt}+U_{xx}-\tfrac13U_{xxxx}.\label{eq:bq2}\eeq
We integrate \eqref{eq:bq2} w.r.t.\ $x$ and $t$, for $-R<x<R$ and $-R<t<R$, with $R$ large, but finite, with a view to later letting $R\to\infty$.
This gives
\begin{align}
\dintR u^2(x,t)\,\d x\,\d t &= \dintR \left\{ U_{tt}(x,t) +U_{xx}(x,t)-\tfrac13 U_{xxxx}(x,t)\right\}\d x\,\d t.\nonumber
\end{align}
The rationale for considering $(x,t)\in[-R,R]\times[-R,R]$, for $R$ large but finite, rather than considering $(x,t)\in\mathbb{R}^2$ from the outset is that for the rational solutions given in \S\ref{sec2}, $U_{xx}(x,t)\not\in L^1(\mathbb{R}^2)$ and $U_{tt}(x,t)\not\in L^1(\mathbb{R}^2)$, though $U_{xx}(x,t)+U_{tt}(x,t)\in L^1(\mathbb{R}^2)$.

\begin{figure}
\[ \begin{array}{c@{\quad}c@{\quad}c} 
\includegraphics[width=5cm]{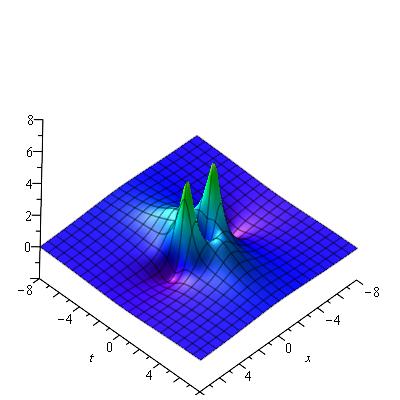}& \includegraphics[width=5cm]{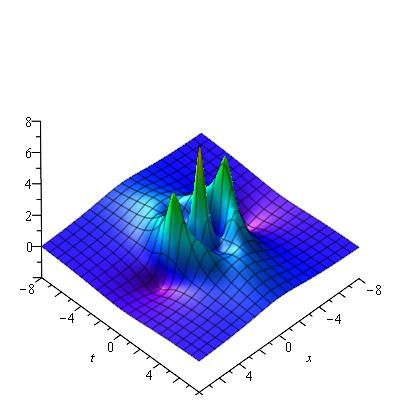}& \includegraphics[width=5cm]{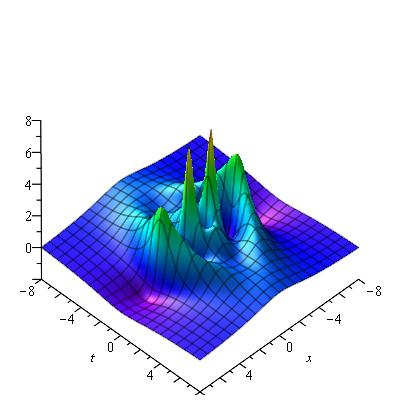}\\
u_2(x,t;0,0) &  u_3(x,t;0,0)   & u_4(x,t;0,0)\\
\includegraphics[width=5cm]{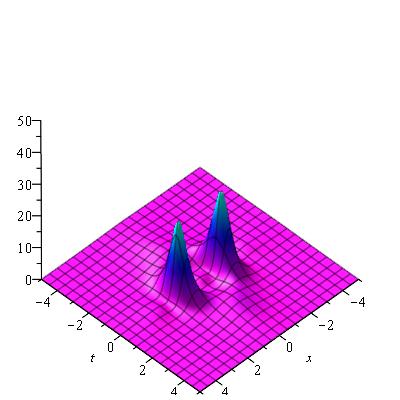}& \includegraphics[width=5cm]{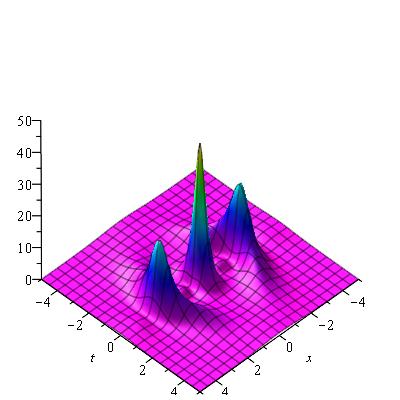}& \includegraphics[width=5cm]{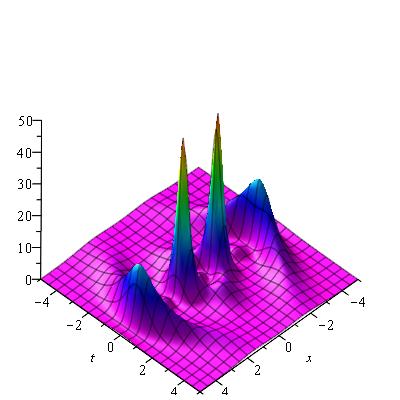}\\
u_2^2(x,t;0,0) &  u_3^2(x,t;0,0)   & u_4^2(x,t;0,0) \\
\includegraphics[width=5cm]{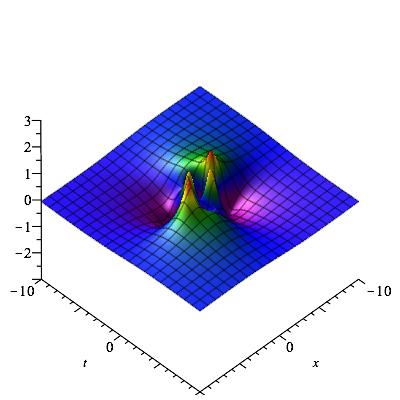}&
 \includegraphics[width=5cm]{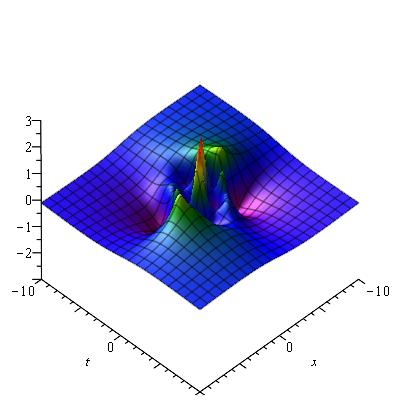}& \includegraphics[width=5cm]{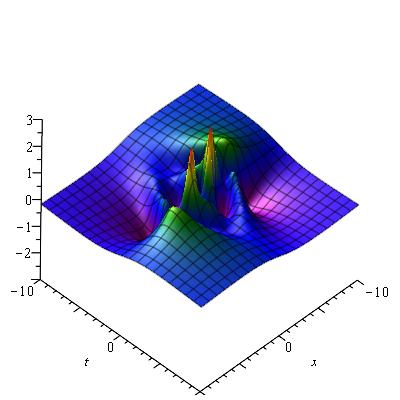}\\
v_2(x,t;0,0) &  v_3(x,t;0,0)   & v_4(x,t;0,0)  \\
\includegraphics[width=5cm]{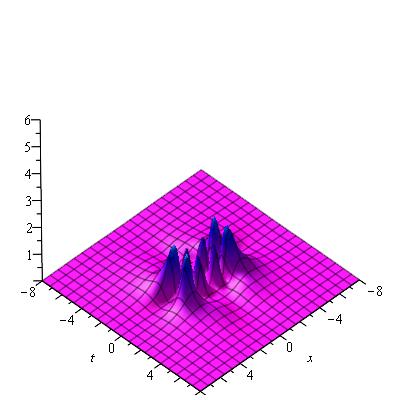}& \includegraphics[width=5cm]{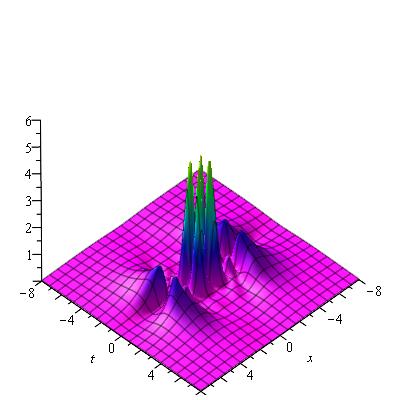}& \includegraphics[width=5cm]{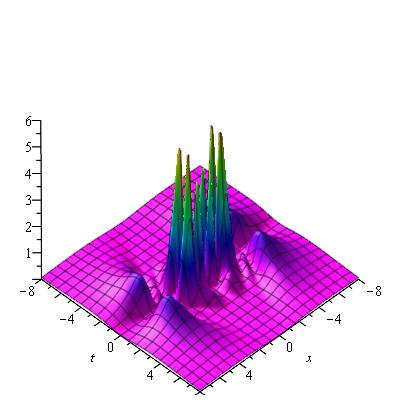}\\
v_2^2(x,t;0,0) &  v_3^2(x,t;0,0)   & v_4^2(x,t;0,0) 
\end{array}\]
\caption{\label{fig41}Plots of the solutions ${u}_j(x,t;0,0)$, $u_j^2(x,t;0,0)$, ${v}_j(x,t;0,0)$ and $v_j^2(x,t;0,0)$ for $j=2,3,4$.}
\end{figure}

For the rational solutions described in \S\ref{sec2}, $U(x,t) =2\ln F_n(x,t)$, where 
\beq F_n(x,t)= (x^2+t^2)^{n(n+1)/2}+G_n(x,t),\label{def:Fn}\eeq
with $G_n(x,t)$ a polynomial of degree $
(n+2)(n-1)$ in both $x$ and $t$. Therefore
\begin{align*}
\frac{1}{8\pi}\dintR \left\{U_{xx} +U_{tt}\right\}\d x\,\d t &=\frac{1}{4\pi}\dintR
\left\{ \left(\frac{F_{n,x}(x,t)}{F_n(x,t)}\right)_x+\left(\frac{F_{n,t}(x,t)}{F_n(x,t)}\right)_t\right\}\d x\,\d t\\
&=\frac{1}{4\pi}\intR\left\{ \frac{F_{n,x}(R,t)}{F_n(R,t)}-\frac{F_{n,t}(-R,t)}{F_n(-R,t)}\right\}\d t\\&\qquad
+\frac{1}{4\pi}\intR\left\{ \frac{F_{n,x}(x,R)}{F_n(x,R)}-\frac{F_{n,t}(x,-R)}{F_n(x,-R)}\right\}\d x,
\end{align*}
where the order of integration has been reversed for the second integral.
Next consider
\begin{align*}
\frac{F_{n,x}(R,t)}{F_n(R,t)} &= \frac{n(n+1)R(R^2+t^2)^{n(n+1)/2-1}+G_{n,x}(R,t)}{(R^2+t^2)^{n(n+1)/2}+G_n(R,t)}\\
&= \frac{n(n+1)R}{R^2+t^2} \left\{1+ \frac{G_{n,x}(R,t)}{(R^2+t^2)^{n(n+1)/2}}\right\}^{-1} 
+\frac{G_{n,x}(R,t)}{(R^2+t^2)^{n(n+1)/2}+G_n(R,t)},
\end{align*}
then letting $t=\tau R$ gives 
\begin{align*}
\frac{1}{4\pi}\intR \left\{ \frac{F_{n,x}(R,t)}{F_n(R,t)}-\frac{F_{n,x}(-R,t)}{F_n(-R,t)}\right\}\d t &=\frac{R}{4\pi}\int_{-1}^{1}
\left\{ \frac{F_{n,x}(R,R\tau)}{F_n(R,R\tau)}-\frac{F_{n,x}(-R,R\tau)}{F_n(-R,R\tau)}\right\}\d \tau\\ 
&=\frac{n(n+1)}{2\pi}\int_{-1}^{1} \frac{1}{1+\tau^2}\left\{1+\O\left(R^{-2}\right)\right\}\d \tau\\
&=\tfrac14{n(n+1)}\left\{1+\O\left(R^{-2}\right)\right\},
\end{align*} since \[\int_{-1}^{1} \frac{\d\,\tau}{1+\tau^2}=\big[\arctan(\tau)\big]_{-1}^1=\tfrac12\pi.\]
Similarly
\begin{align*}
\frac{F_{n,t}(x,R)}{F_n(x,R)}
&= \frac{n(n+1)R}{x^2+R^2} \left\{1+ \frac{G_{n,t}(x,R)}{(x^2+R^2)^{n(n+1)/2}}\right\}^{-1} 
+\frac{G_{n,t}(x,R)}{(x^2+R^2)^{n(n+1)/2}+G_n(x,R)}
\end{align*}
and letting $x=\xi R$ gives 
\begin{align*}
\frac{1}{4\pi}\intR \left\{ \frac{F_{n,t}(x,R)}{F_n(x,R)}-\frac{F_{n,t}(x,-R)}{F_n(x,-R)}\right\}\d x 
&=\frac{n(n+1)}{2\pi}\int_{-1}^{1} \frac{1}{1+\xi ^2}\left\{1+\O\left(R^{-2}\right)\right\}\d \xi \\
&=\tfrac14{n(n+1)}\left\{1+\O\left(R^{-2}\right)\right\}
\end{align*}
Hence we have shown that
\[\frac{1}{8\pi}\dintR \left\{u^2(x,t)+\tfrac13U_{xxxx}(x,t)\right\} \d x\,\d t=\tfrac12{n(n+1)}\left\{1+\O\left(R^{-2}\right)\right\},\]
We note that \[\lim_{R\to\infty}\dintR U_{xxxx}(x,t)\,\d x\,\d t=0,\]
since $\lim_{|R|\to\infty}U_{xxx}(R,t)=0$, and so in the limit as $R\to\infty$ we obtain the following result. 
\begin{theorem}\label{ref:thm1}
If $u_n(x,t)$ is an algebraically decaying rational solution of the Boussinesq equation \eqref{eq:bq} given by \eqref{ungen}, then
\beq\frac{1}{8\pi}\dimp u_n^2(x,t)\,\d x\,\d t =\tfrac12{n(n+1)}.\label{thm41}\eeq
\end{theorem}
For $n=2$, the result \eqref{thm41} equals the number of peaks, i.e.\ $3$, as shown in Figure \ref{fig21}, and for $n=3$, it equals the number of peaks, i.e.\  $6$, as shown in Figure \ref{fig22}.

In Figure \ref{fig41}, plots are given of the solutions ${u}_j(x,t)$, $u_j^2(x,t)$, ${v}_j(x,t)$ and $v_j^2(x,t)$ for $j=2,3,4$, with $\a=\b=0$.

\subsection{Integral of $u_n^3(x,t)$.}
Now we shall consider the integral of $u_n^3(x,t)$, i.e.\ result \eqref{int:un3}. From the third conservation law
\beq(uv)_t+\left(\tfrac23 u^3+\tfrac12 v^2-\tfrac12u^2-\tfrac16 u_x^2+\tfrac13 uu_{xx}\right)_{x}=0,\label{conlaw3}\eeq
we have
\[\imp\left(\tfrac23 u^3+\tfrac12 v^2-\tfrac12 u^2-\tfrac16 u_x^2+\tfrac13 uu_{xx}\right) \d t=0;\]
recall \eqref{bqcom23} with $k_3=0$. Integrating this result w.r.t.\ $x$ and interchanging the order of integration gives
\begin{align}\dimp u^3\,\d x\,\d t &= \dimp\left(\tfrac34 u^2-\tfrac34 v^2+\tfrac14 u_x^2-\tfrac12 uu_{xx}\right) \d x\,\d t.\nonumber\end{align}
One integration by parts yields
\begin{align}\dimp u^3\,\d x\,\d t&= \tfrac34\dimp\left( u^2- v^2+ u_x^2\right)\d x\,\d t.\label{eq3}
\end{align}
From the fourth conservation law
\beq\left(\tfrac23 u^3+v^2-u^2-\tfrac13u_x^2\right)_t+
\left(2u^2v-2uv+\tfrac23 vu_{xx}-\tfrac23 u_xv_x\right)_x=0,\label{conlaw4}\eeq
we have
\[\imp \left(\tfrac23 u^3+v^2-u^2-\tfrac13u_x^2\right)\d x=0;\]
recall \eqref{bqcom14} with $c_4=0$. Integrating this w.r.t.\ $t$ gives
\begin{align}\dimp u^3\,\d x\,\d t &= \dimp\left(\tfrac32 u^2-\tfrac32 v^2+\tfrac12u_x^2\right)\d x\,\d t.\label{eq4}
\end{align}
Therefore equations \eqref{eq3} and \eqref{eq4} give
\begin{subequations}\label{eq34ab}\beq\dimp u^3(x,t)\,\d x\,\d t = \dimp u_x^2(x,t)\,\d x\,\d t, \label{eq34a}
\eeq and \beq
\dimp u^3(x,t)\,\d x\,\d t=3 \dimp\left\{ u^2(x,t)-v^2(x,t)\right\}\d x\,\d t.\label{eq34b} \eeq\end{subequations}
From the Boussinesq equation \eqref{eq:bq}, we have
\[ u^2=u+U_{tt}-\tfrac13u_{xx},\]
where $u=U_{xx}$, so
\begin{align*}\dimp u^3\,\d x\,\d t &= \dimp\left(u^2+u\,U_{tt}-\tfrac13 uu_{xx}\right)\d x\,\d t.\end{align*}
Using integration by parts gives
\begin{align}\dimp u^3\,\d x\,\d t&= \dimp\left(u^2+u\,U_{tt}+\tfrac13 u_{x}^2\right)\d x\,\d t,
\end{align}
and so from \eqref{eq34ab} we obtain
\begin{subequations}\label{eq11ab}\beq\dimp u^3\,\d x\,\d t =\tfrac32\dimp\left(u^2+u\,U_{tt}\right)\d x\,\d t, \eeq 
and 
\beq\dimp u^3\,\d x\,\d t= 3\dimp\left(v^2+u\,U_{tt}\right)\d x\,\d t.\eeq\end{subequations}
Consequently, we see that
\beq \dimp u^2\,\d x\,\d t=\dimp \left(2v^2+u\,U_{tt}\right)\d x\,\d t.\label{eq10b}\eeq

\begin{lemma} \label{ref:lem1} Suppose that $u(x,t)$ and $v(x,t)$ are solutions of the system \eqref{eq:bqs}, and $u(x,t)=U_{xx}(x,t)$ with
\[\lim_{|x|\to\infty}U_x(x,t)=0,\qquad\lim_{|t|\to\infty}U_x(x,t)=0,\] then 
\beq \dimp v^2(x,t)\,\d x\,\d t = \dimp u(x,t)\,U_{tt}(x,t)\,\d x\,\d t .\label{eq11}\eeq
\end{lemma}

\begin{proof}
Since $u=U_{xx}$ and $v=-U_{xt}$ then
\[\dimp v^2(x,t)\,\d x\,\d t = \imp\left(\imp U_{xt}^2\,\d t\right)\d x =-\imp\left(\imp U_xU_{xtt}\,\d t\right)\d x,\]
 and
\[\dimp u(x,t)\,U_{tt}(x,t)\,\d x\,\d t =\imp\left(\imp U_{xx}\,U_{tt}\,\d x\right)\d t =-\imp\left(\imp U_xU_{xtt}\,\d x\right)\d t,\]
so the result follows, since the order of integration can be switched.
\end{proof}

Consequently we have the following result.
\begin{theorem}\label{ref:thm2}
If $u_n(x,t)$ is an algebraically decaying rational solution of the Boussinesq equation \eqref{eq:bq} given by \eqref{ungen}, then
\beq \frac{1}{8\pi}\dimp u_n^3(x,t)\,\d x\,\d t=n(n+1).\eeq
\end{theorem}
\begin{proof}
From equation \eqref{eq10b} and Lemma \ref{ref:lem1} we see that
\beq\dimp u^2\,\d x\,\d t=\dimp \left(2v^2+u\,U_{tt}\right)\,\d x\,\d t=3\dimp v^2\,\d x\,\d t.\label{eq15}\eeq
Then from equations \eqref{eq34b} and \eqref{eq15} we have
\beq\dimp u^3\,\d x\,\d t =2\dimp u^2\,\d x\,\d t,\eeq
and so using Theorem \ref{ref:thm1} we obtain the result.
\end{proof}

\begin{corollary}\label{ref:lem2}
If $v_n(x,t)$ is an algebraically decaying rational solution of the system \eqref{eq:bqs}, then
\beq \frac{1}{8\pi} \dimp v_n^2(x,t)\,\d x\,\d t =\tfrac16n(n+1)
.\label{eq14}
\eeq
\end{corollary}
\begin{proof}The result follows immediately from equation \eqref{eq15} and Theorem \ref{ref:thm1}.
\end{proof}

\subsection{Integrals of $u_1^m(x,t)$.}
Finally we consider the integral of $u_1^m(x,t)$, for $m\geq2$. 
\begin{theorem} Consider the rational solution of the Boussinesq equation given by 
\beq u_1(x,t)=2\pderiv[2]{}{x}\ln(x^2+t^2+1)=\frac{4(1-x^2+t^2)}{(1+x^2+t^2)^2}.\eeq Then 
\beq \label{int:u1m}\frac{1}{8\pi}\dimp u_1^m(x,t)\,\d x\,\d t=\frac{m!}{(2m-1)!}\sum_{\ell=0}^{\lfloor m/2\rfloor}\frac{(2\ell)!(2m-2\ell-2)!}{2^{2\ell-2m-3}(\ell !)^2(m-2\ell)!},\eeq where $\lfloor x\rfloor$ is the largest integer less than or equal to $x$, for $m$ an integer with $m\geq2$.
\end{theorem}
\begin{proof}
We need to evaluate 
\[\dimp u_1^m(x,t)\,\d x\,\d t=2^{2m}\int_{-\infty}^{\infty}\int_{-\infty}^{\infty} \frac{(1-x^2+t^2)^m}{(1+x^2+t^2)^{2m}}\,\d x\,\d t,\]
for all integers $m\geq 2$. If we make the transformation $x=\sqrt{\rho}\,\cos\ph$, $t=\sqrt{\rho}\,\sin\ph$ then
\[\dimp u_1^m(x,t)\,\d x\,\d t =2^{2m-1}\int_0^\infty\int_0^{2\pi}\sum_{k=0}^m
 (-1)^k{m \choose k}\frac{\rho^k\cos^k(2\ph) }{(1+\rho)^{2m}} \,\d\ph\,\d\rho.\]
 Elementary results give that 
\[\int_0^{2\pi}\cos^{2\ell}(2\ph)\,\d\ph=\frac{2\sqrt{\pi}\,\Gamma(\ell+\tfrac12)}{\Gamma(\ell+1)}
\equiv\frac{\pi}{2^{2\ell -1}}\frac{ (2\ell)!}{(\ell !)^2},\qquad\qquad \int_0^{2\pi}\cos^{2\ell+1}(2\ph)\,\d\ph=0,\]
for integer $\ell$, which, when combined with the knowledge that
\[\int_0^{\infty}\frac{\rho^k}{(1+\rho)^{2m}}\,\d\rho=\frac{\Gamma(k+1)\,\Gamma(2m-k-1)}{\Gamma(2m)}
\equiv\frac{ k!(2m-k-2)!}{(2m-1)!},\qquad{\rm if}\quad k<2m-1,\]
 leads to the desired result \eqref{int:u1m}.
\end{proof}

\begin{remark}{\rm
We note that for the rational solution $u_1(x,t)$
\[\frac{1}{8\pi}\imp u_1(x,t)\,\d x= \frac{1}{2\pi}\imp \frac{1-x^2+t^2}{(1+x^2+t^2)^2}\,\d x=0,\]
and therefore
\[\frac{1}{8\pi}\imp\left(\imp u_1(x,t)\,\d x\right)\d t= 0.\]
On the other hand
\[\frac{1}{8\pi}\imp u_1(x,t)\,\d t = \frac{1}{2\pi}\imp \frac{1-x^2+t^2}{(1+x^2+t^2)^2}\,\d t= \frac{1}{2(x^2+1)^{3/2}},\]
and then 
\[\frac{1}{8\pi}\imp\left(\imp u_1(x,t)\,\d t\right)\d x= 1.\]
So for $u_1(x,t)$ the order of integration is important since $u_1(x,t)\not\in L^1(\mathbb{R}^2)$. In fact more generally $u_n(x,t)\not\in L^1(\mathbb{R}^2)$.}\end{remark}

\section{\label{seckp}Rational solutions of the \kp\ I equation}
The \kp\ (KP) equation
\begin{equation}\label{eq:kp}
(v_\tau+6vv_\xi+v_{\xi\xi\xi})_\xi+3\sigma^2 v_{\eta\eta}=0, \qquad \sigma^2=\pm1,
\end{equation} 
which is known as KPI if $\sigma^2=-1$, i.e.\  \eqref{eq:kp1}, 
and KPII if $\sigma^2=1$, was derived by Kadomtsev and Petviashvili \cite{refKP} to model ion-acoustic waves of small amplitude propagating in plasmas and is a two-dimensional generalisation of the KdV equation \eqref{eq:kdv}.
The KP equation arises in many physical applications including weakly two-dimen\-sion\-al long waves in shallow water \cite{refAS79,refSegurF}, where the sign of $\sigma^2$ depends upon the relevant magnitudes of gravity and surface tension.
The KP equation \eqref{eq:kp} is also a completely integrable soliton equation solvable by inverse scattering and again the sign of $\sigma^2$ is critical since if $\sigma^2=-1$, then the inverse scattering problem is formulated in terms of a Riemann-Hilbert problem \cite{refFA83,refManakov}, whereas for $\sigma^2=1$, it is formulated in terms of a $\overline\partial$ (``DBAR'') problem \cite{refABF}. 

The first rational solution of KPI \eqref{eq:kp1},  is the so-called ``lump solution"
\begin{equation}\label{kp:lump}v(\xi,\eta,\tau)=2\pderiv[2]{}{\xi}\ln [(\xi-3\tau)^2+\eta^2+1] 
= -4\frac{(\xi-3\tau)^2-\eta^2-1}{[(\xi-3\tau)^2+\eta^2+1]^2},\end{equation}
which was found by Manakov \etal\ \cite{refMZBIM}. Subsequent studies of rational solutions, also known as lump solutions, of KPI \eqref{eq:kp1} include Ablowitz \etal\ \cite{refACTV},  Ablowitz and Villarroel \cite{refAV,refVA99}, Dubard and Matveev \cite{refDubMat11,refDubMat13}, Gaillard \cite{refGail16a,refGail16b}, Johnson and Thompson \cite{refJT}, 
Ma \cite{refMa}, Pelinovsky \cite{refPel94,refPel98}, Pelinovsky and Stepanyants \cite{refPelStep}, 
Satsuma and Ablowitz \cite{refSatAb79}, and Singh and Stepanyants \cite{refSS}.

Dubard and Matveev \cite{refDubMat11,refDubMat13} derive rational solutions of KPI \eqref{eq:kp1}
from generalised rational solutions of the focusing NLS equation \eqref{eq:fnls} which are discussed in Appendix A below; see also \cite{refDGKM,refGail16a,refGail16b}. 
The Boussinesq equation \eqref{eq:bq} is a symmetry reduction of KPI \eqref{eq:kp1} and so the rational solutions  
of the Boussinesq equation can be used to generate rational solutions of KPI.
If in KPI \eqref{eq:kp1} we make the travelling wave reduction
\[ v(\xi,\eta,\tau)=u(x,t),\qquad x=\xi-3\tau,\quad t=\eta,\]
then $u(x,t)$ satisfies the Bouss\-inesq equation \eqref{eq:bq}.
Consequently given a solution of the Bouss\-inesq equation \eqref{eq:bq}, then we can derive a solution of KPI \eqref{eq:kp1}. In particular, if 
\[ u(x,t)=2\pderiv[2]{}{x}\ln F(x,t),\]
for some known $F(x,t)$, is a solution of the Bouss\-inesq equation \eqref{eq:bq}, then
\[ v(\xi,\eta,\tau)=2\pderiv[2]{}{\xi}\ln F(\xi-3\tau,\eta),\]
is a solution of KPI \eqref{eq:kp1}. For example the choice $F(x,t)=x^2+t^2+1$ gives the lump solution \eqref{kp:lump} of KPI.  In \cite{refCD16} Clarkson and Dowie show that the family of rational solutions for KPI \eqref{eq:kp1} derived through this approach are fundamentally different from those derived by Dubard and Matveev \cite{refDubMat11,refDubMat13} from the rational solutions of the  focusing NLS equation \eqref{eq:fnls}.

Ablowitz \etal\ \cite{refACTV,refAV,refVA99} derived a hierarchy of algebraically decaying rational solutions, or  lump solutions, of KPI \eqref{eq:kp1} which have the form
\begin{equation}\label{KPi:ACTV}
v_m(\xi,\eta,\tau)=2\pderiv[2]{}{\xi} \ln G_m(\xi,\eta,\tau),\end{equation}
where $G_m(\xi,\eta,\tau)$ is a polynomial of degree $2m$ in $\xi$, $\eta$ and $\tau$. These rational solutions are derived in terms of the eigenfunctions of the non-stationary \sch\ equation
\begin{equation}\label{KPist1}
\i\varphi_{\eta}+\varphi_{\xi\xi}+v\varphi=0,
\end{equation}
with potential $v=v(\xi,\eta,\tau)$, which is used in the solution of KPI \eqref{eq:kp1} by inverse scattering; equation \eqref{eq:kp1} is obtained from the compatibility of \eqref{KPist1} and 
\begin{equation}\label{KPist2}
\varphi_{\tau}+4\varphi_{\xi\xi\xi}+6v\varphi_{\xi}+w\varphi=0,\qquad w_{\xi}=v.
\end{equation}
This is a different hierarchy of rational solutions of KPI \eqref{eq:kp1} compared to those discussed above, not least because it involves polynomials of \textit{all} even degrees, not just of degree $n(n+1)$, with $n\in\mathbb{N}$. These rational solutions of KPI \eqref{eq:kp1} are deeply connected with an integer called the ``charge" or ``index", and this number is related to the degree of the polynomial that generates the rational solution \cite{refACTV,refAV,refVA99}. It would be interesting to investigate whether there is an analogous result for the rational solutions of the Bouss\-inesq equation \eqref{eq:bq} discussed here which might explain the results. However we shall not pursue this further here.

\comment{which is either n or closely related to the integer n used in this paper. And the generating polynomial satisfies the same eq (2.2).  It may well be that the solutions of the Boussinesq eq are not only related to those of the KP eq but the connection to the KP eq can further elucidate deep underlying properties of the rational solutions the authors discuss, in particular the charge/index.}

\section{\label{sec5}Discussion}
Amongst his extensive contributions to fluid mechanics, David Benney conducted many studies of nonlinear wave equations. Long waves were the topic of interest in \cite{refBenney73,refBenney06,refBR,refYB} while lump solutions of the modified Zakharov-Kuznetsov equation are studied in \cite{refSB}. A recurring theme in Benney's work was that of conservation laws and in this spirit here 
we have been concerned with conservation laws and integral relations associated with algebraically decaying rational solution of the Boussinesq equation \eqref{eq:bq} given by \eqref{ungen}. In addition to the results discussed above, we have performed numerical investigations of higher integral relations, i.e.\ 
\[ \frac{1}{8\pi}\int_{-\infty}^{\infty} \int_{-\infty}^{\infty} u_n^m(x,t;\a,\b)\,\d x\,\d t,\] for $m\geq4$, where $u_n(x,t;\a,\b)$ is an algebraically decaying rational solution of the Boussinesq equation given by \eqref{ungen}, with $n\geq2$. However for $m\geq4$ there appears to be no pattern analogous to the results for $m=2$ and $m=3$ given in Theorem \ref{thm1}.

It is interesting to compare the results described here with analogous conservation laws and integral relations for rogue wave solutions of the NLS equation \eqref{eq:fnls}. It is straightforward to show that the first few constants of motion and associated fluxes for rogue wave solutions of NLS equation \eqref{eq:fnls} are zero, which was the case for the Boussinesq equation \eqref{eq:bq}; 
the first few rogue wave solutions and conservation laws for the NLS equation \eqref{eq:fnls} are given in Appendix A below.
Ankiewicz and Akhmediev \cite{refAA15} conjectured that the number of components in any NLS rogue wave is a triangle number which can be calculated as the integral of the squared deviation from the background level across the space-time plane.

\begin{conjecture}{Suppose that $\psi_n(x,t)$ 
is a rogue wave 
solution of the focusing NLS equation \eqref{eq:fnls}
then
\beq Q_n= \frac{1}{32\pi}\int_{-\infty}^{\infty} \int_{-\infty}^{\infty} \left[|\psi_n(x,t)|^{2}-1\right]^2\d x\,\d t=\tfrac12n(n+1).\eeq
}\end{conjecture}

This conjecture is the analogue of \eqref{int:un2}; the factor of $4$ is due to the fact that Ankiewicz and Akhmediev \cite{refAA15} consider an NLS equation which is obtained from \eqref{eq:fnls} by letting $t\to \tfrac12x$ and $x\to \tfrac12t$.

\begin{figure}
\[ \begin{array}{c@{\quad}c} 
\includegraphics[width=7.5cm]{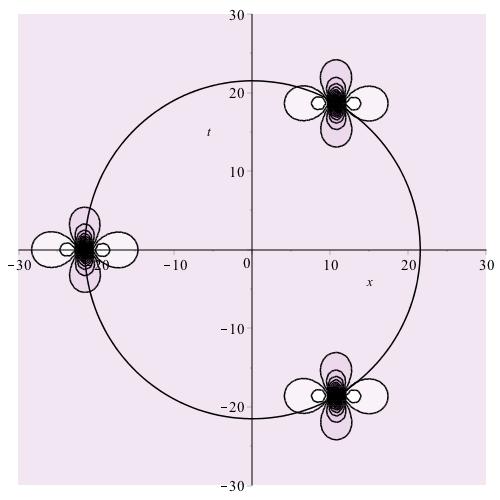}& \includegraphics[width=7.5cm]{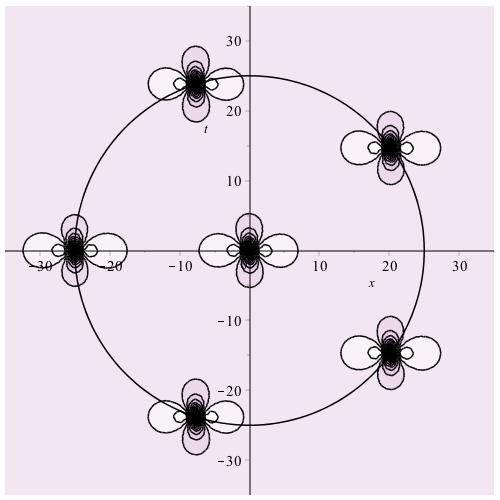}\\
\uu_2(x,t;0,10^4) &  \uu_3(x,t;0,10^7)  
\end{array}\]
\caption{\label{fig51}Contour plots of the solutions $\uu_2(x,t;0,10^4)$ and $\uu_3(x,t;0,10^7)$ of the Boussinesq equation.}
\end{figure}

\begin{figure}
\[ \begin{array}{c@{\quad}c} 
\includegraphics[width=7.5cm]{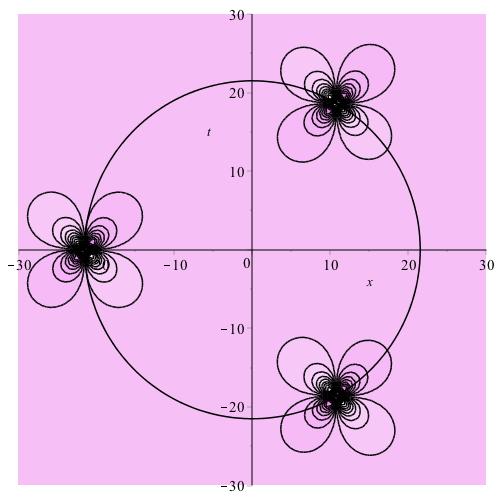}& \includegraphics[width=7.5cm]{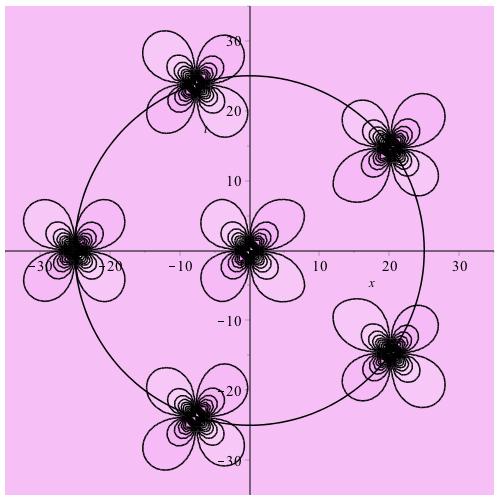}\\
v_2(x,t;0,10^4) & v_3(x,t;0,10^7)  
\end{array}\]
\caption{\label{fig52}Contour plots of the solutions $v_2(x,t;0,10^4)$ and $v_3(x,t;0,10^7)$ of the Boussinesq system.}
\end{figure}

We note that if the parameters $\a$ and $\b$ sufficiently large, then the rational solution $\uu_2(x,t;\a,\b)$ has three lumps which are essentially copies of the lowest-order solution, of approximately the same height and equally spaced on a circle. An analogous situation arises for the second generalised rational solution of the NLS equation \cite{refKAA12}.
Further, for $\a$ and $\b$ sufficiently large, the rational solution $\uu_3(x,t;\a,\b)$ has six lumps, again essentially copies of the lowest-order solution, which are of approximately the same height and with five of these equally spaced on a circle. Again an analogous situation arises for the third generalised rational solution of the NLS equation \cite{refKAA11}.
For both cases we conjecture that the radius of the circle is equal to $(\a^2+\b^2)^{1/h}$, for some $h$ which depends on $n$, as appears to be the case for the NLS equation \cite{refKAA11,refKAA12}, though we shall not investigate this further here.
Contour plots of the solutions $\uu_2(x,t;0,\b)$ and $\uu_3(x,t;0,\b)$, for $\b=10^4$ and $\b=10^7$ respectively,  of the Boussinesq equation \eqref{eq:bq} illustrating this behaviour are given in Figure \ref{fig51}.
Contour plots of $v_2(x,t;0,\b)$ and $v_3(x,t;0,\b)$, for $\b=10^4$ and $\b=10^7$ respectively, where
\[ v_n(x,t;\a,\b)=-2\frac{\partial^2}{\partial x\partial t}\ln F_n(x,t;\a,\b),\]
with $F_{n}(x,t;\a,\b)$ given by \eqref{bq:untilde} are given in Figure \ref{fig52}. 

Rogue wave solutions have also been derived for $(2+1)$-dimensional equations such as the Benney-Roskes equation \cite{refBR}, also known as the Davey-Stewartson equation \cite{refDS} (see also \cite{refAC,refAS81})
\begin{subequations}\label{eq:DS}\begin{align} 
&\mbox{$\i$}q_t= q_{xx}+\sigma^2 q_{yy}+(\ep |q|^2-2\phi)q,\\
&\phi_{xx}-\sigma^2 \phi_{yy}=\ep(|q|^2)_{xx},
\end{align}\end{subequations}
where $\sigma^2=\pm1$ and $\ep=\pm1$, independently; see \cite{refOY12,refOY13} for details of rogue wave solutions of \eqref{eq:DS}. It would be interesting to see if there are analogous results to those given in this paper for $(2+1)$-dimensional equations such as equation \eqref{eq:DS}.


\subsection*{Acknowledgments}
AA acknowledges the support of the Australian Research Council (Discovery Project number DP140100265) and the Volkswagen Stiftung.
PAC and ED thank the School of Mathematics \& Physics at the University of Tasmania, Hobart, Australia, for their hospitality during their visit when some of this research was done. 
We thank the reviewer for their helpful comments.

\appendix
\def\ps{\psi^*}
\section{Rational solutions and conservation laws for the focusing NLS equation}
Rational solutions of the focusing NLS equation \eqref{eq:fnls} have the general form
\begin{equation}
\psi_n(x,t)=\left\{1-4\frac{G_n(x,t)+\i tH_n(x,t)}{{D}_n(x,t)}\right\}\exp\left(\tfrac12\i t\right),
\end{equation}
where $G_n(x,t)$ and $H_n(x,t)$ are polynomials of degree $(n+2)(n-1)$ in both $x$ and $t$, with total degree $(n+2)(n-1)$, and ${D}_n(x,t)$ is a polynomial of degree $n(n+1)$ in both $x$ and $t$, with total degree $n(n+1)$ and has no real zeros. \comment{The polynomials $D_n(x,t)$, $G_n(x,t)$ and  $H_n(x,t)$ satisfy the Hirota equations
\begin{align*}
&4(t\D_t+1) H_n\cdot D_n+\D_x^2 D_n\cdot D_n - 4\D_x^2D_n\cdot G_n=0,\\
&\D_t G_n\cdot D_n+t\D_x^2H_n\cdot D_n=0,\\
&\D_x^2D_n\cdot D_n=8G_n^2+8t^2H_n^2-4 D_nG_n.
\end{align*}}%
The first two rational solutions of the focusing NLS equation \eqref{eq:fnls} have the form \cite{refAASG,refACA}
\begin{align}
\psi_1(x,t)&=\left\{1-\frac{4(1+\i t)}{x^2+t^2+1}\right\}\exp\left(\tfrac12\i t\right),\label{nls:psi1}\\
\psi_2(x,t)&=\left\{1-12\,\frac{G_2(x,t)+\i tH_2(x,t)}{{D}_2(x,t)}\right\}\exp\left(\tfrac12\i t\right),\label{nls:psi2}
\end{align}
where 
\begin{align*}
G_2(x,t)&= x^4 + 6(t^2+1)x^2+ 5t^4+18t^2-3,\\
H_2(x,t)&= x^4+2(t^2-3)x^2+ (t^2+5)(t^2-3),\\
{D}_2(x,t)&= x^6+3(t^2+1)x^4 + 3(t^2-3)^2x^2 + t^6+27t^4+99t^2+9. \label{nls:f2}
\end{align*}
Further 
\begin{equation}|\psi_n(x,t)|^2=1+4\pderiv[2]{}{x}\ln D_n(x,t).\end{equation}

Dubard \etal\ \cite{refDGKM} show that the rational solutions of the focusing NLS equation \eqref{eq:fnls} can be generalised to include some arbitrary parameters. The first of these generalised solutions has the form
\begin{align}\label{nls:u2hat}\widehat{\psi}_2(x,t;\a,\b)=
\left\{1-12\frac{\widehat{G}_2(x,t;\a,\b)+\i \widehat{H}_2(x,t;\a,\b)}{\widehat{D}_2(x,t;\a,\b)}\right\}\exp\left(\tfrac12\i t\right),\end{align}
where
\begin{align*}
\widehat{G}_2(x,t;\a,\b)&= 
G_2(x,t) - 2\a t + 2\b x,\\
\widehat{H}_2(x,t;\a,\b)&= 
tH_2(x,t)+ \a(x^2-t^2+1) + 2\b xt,\\
\widehat{D}_2(x,t;\a,\b)&= 
{D}_2(x,t) + 2\a t(3x^2-t^2-9) - 2\b x(x^2-3t^2-3)+\a^2+\b^2,\label{nls:F2}
\end{align*}
with $\a$ and $\b$ arbitrary constants, see also \cite{refAKA,refDubMat11,refDubMat13,refGail11,refGLL,refKAA11,refKAA12,refKAA13,refOY}.

The first few conservation laws for the focusing NLS equation \eqref{eq:fnls} are
\begin{subequations}\label{nlsconlaws}\begin{align}
&(|\psi|^2-1)_t+\i(\psi\ps_x-\ps\psi_x)_x=0,\\
&(\psi\ps_x-\ps\psi_x)_t+\i\left((\ps\psi_{xx}+\psi\ps_{xx})+\tfrac12|\psi|^4-2\psi_x\ps_x\right)_x=0,\\
&\left(4\psi_x\ps_x-|\psi|^4+1\right)_t+2\i\left(|\psi|^2(\ps\psi_x-\psi\ps_x)+2(\psi_x\ps_{xx}-\ps_x\psi_{xx}) \right)_x=0.
\end{align}\end{subequations}
We remark that the conserved quantities in \eqref{nlsconlaws} appear in \cite{refYB}, see equations (4.21)--(4.23) in that paper.

\def\OUP{O.U.P.} 
\def\CUP{C.U.P.} 

\def\refpp#1#2#3#4{\vspace{-0.2cm}
\bibitem{#1} \textrm{\frenchspacing#2}, \textrm{#3}, #4.}

\def\refjl#1#2#3#4#5#6#7{\vspace{-0.2cm}
\bibitem{#1}{\frenchspacing\rm#2}, {\rm#6}, 
\textit{\frenchspacing#3}, \textbf{#4}\ (#7)\ #5.}

\def\refbk#1#2#3#4#5{\vspace{-0.2cm}
\bibitem{#1}{\frenchspacing\rm#2}, ``\textit{#3}," #4 (#5).} 

\def\refcf#1#2#3#4#5#6{\vspace{-0.2cm}
\bibitem{#1} \textrm{\frenchspacing#2}, \textrm{#3},
in ``\textit{#4}" [{\frenchspacing#5}], #6.}

\end{document}